\documentclass[12pt]{article}
\usepackage[T1]{fontenc}
\usepackage{kpfonts}
\usepackage[dvipsnames]{xcolor}
\usepackage{pdfpages}
\usepackage{sgame}
\usepackage{accents}
\usepackage{tikz-cd}
\usepackage{float}
\usepackage[round]{natbib}   
\usepackage{multirow}
\usepackage{dutchcal}
\usepackage{pdfpages}
\usepackage{bbm}
\usepackage{sgame}
\usepackage{mathtools}
\usepackage{amsthm}
\usepackage{enumitem}
\usepackage[margin=1in]{geometry}
\usepackage{caption}
\usepackage{subcaption}
\usepackage{epigraph}
\usepackage{listings}
\usetikzlibrary{calc}
\usetikzlibrary{shapes,arrows}

\newcommand{\ubar}[1]{\underaccent{\bar}{#1}}

\newtheorem{theorem}{Theorem}[section]
\newtheorem{proposition}[theorem]{Proposition}
\newtheorem{lemma}[theorem]{Lemma}

\newtheorem{claim}[theorem]{Claim}

\theoremstyle{definition}

\newtheorem{definition}[theorem]{Definition}

\newtheorem{assumption}[theorem]{Assumption}
\newtheorem{condition}[theorem]{Condition}

\newtheorem{remark}[theorem]{Remark}

\usepackage{hyperref}
\hypersetup{
    colorlinks=true,
    linkcolor=Sepia,
    filecolor=Sepia,      
    urlcolor=Sepia,
    citecolor=Sepia,
}

\definecolor{backcolour}{rgb}{0.63, 0.79, 0.95}

\lstdefinestyle{mystyle}{
  backgroundcolor=\color{backcolour},
  basicstyle=\ttfamily\footnotesize,
  breakatwhitespace=false,         
  breaklines=true,                 
  captionpos=b,                    
  keepspaces=true,                 
  numbers=left,                    
  numbersep=5pt,                  
  showspaces=false,                
  showstringspaces=false,
  showtabs=false,                  
  tabsize=2
}

\providecommand{\keywords}[1]{\textbf{\textit{Keywords:}} #1}
\providecommand{\jel}[1]{\textbf{\textit{JEL Classifications:}} #1}

\begin{document}
\author{Vasudha Jain\thanks{Indian Institute of Technology Kanpur, Email: \href{mailto:vasudha.jain@utexas.edu}{vasudhaj@iitk.ac.in}} \and Mark Whitmeyer\thanks{Arizona State University, Email: \href{mailto:mark.whitmeyer@gmail.com}{mark.whitmeyer@gmail.com} \newline We thank Gregorio Curello, Claudia Herresthal, Tatiana Mayskaya, Franz Ostrizek, Mallesh Pai, Joseph Whitmeyer and Thomas Wiseman for their comments. We are also grateful to the participants in the Bonn Postdoc Seminar for their feedback. This work was begun while the first author was at UT Austin and the second author was at the Institute for Microeconomics and the Hausdorff Center for Mathematics at the University of Bonn, where he was generously supported by the DFG under Germany's Excellence Strategy-GZ 2047/1, Projekt-ID
390685813.} }

\title{Whose Bias?}

\date{\today}

\maketitle

\begin{abstract} 
Law enforcement acquires costly evidence with the aim of securing the conviction of a defendant, who is convicted if a decision-maker's belief exceeds a certain threshold. Either law enforcement or the decision-maker is biased and is initially overconfident that the defendant is guilty. Although an innocent defendant always prefers an unbiased decision-maker, he may prefer that law enforcement have some bias to none. Nevertheless, fixing the level of bias, an innocent defendant may prefer that the decision-maker, not law enforcement, is biased.
\end{abstract}
\\
\jel{D02, D83, D90, K42} 
\newpage
\setlength{\epigraphwidth}{3.3in}\begin{epigraphs}
\qitem{A fox should not be of the jury at a goose's trial.}%
      {Thomas Fuller}
\end{epigraphs}

\section{Introduction}
How does bias shape law enforcement's information gathering protocol, as it attempts to secure the conviction of a defendant? We explore a simple three party model. There is a defendant, suspected of a crime; law enforcement, who wants the defendant to be convicted; and a decision-maker, who decides whether to convict the defendant based on the evidence provided by law enforcement. 

In our model, the defendant is either guilty or innocent, and the true prior probability of his guilt is $\mu$. We suppose one of either law enforcement or the decision-maker is biased and has an incorrectly high prior probability that the defendant is guilty--either $\mu_L$ or $\mu_{DM}$ is strictly greater than $\mu$. We address three questions: 1. Must this belief-based bias always hurt an innocent defendant and make it more likely that he is (wrongfully) convicted? 2. Does it matter which party, decision-maker or law enforcement, is biased? 3. Fixing the level of bias, whose bias is worse for an innocent defendant?

These are important questions. Indeed, there is an enormous literature that looks at bias in the legal system. \cite{fryer2019empirical} explores racial differences in police use of force and finds that there are significant differences in how police interact with civilians. \cite{devi2020policing} discover that well-publicized deadly encounters lead to a subsequent change in police investigations. These relate to our model in that they are illustrations of law enforcement's ($L$'s) prejudice. A large portion of the literature documents evidence of the decision-maker's ($DM$'s) bias: \cite{shayo2011judicial}, \cite*{anwar2012impact}, \cite*{abrams2012judges}, and \cite{cohen2019judicial} all investigate racial and gender disparities in sentencing.

In our framework, conviction occurs if the decision-maker's posterior belief about the defendant's guilt exceeds an exogenously specified threshold. Law enforcement is single-minded and only wishes to maximize the probability that the defendant is convicted. In order to obtain a conviction, law enforcement may gather evidence at a cost, which we micro-found via a continuous time Wald problem (\cite{wald1945sequential}). Importantly, all evidence must be disclosed, both positive and negative, as law enforcement is prohibited from hiding exculpatory findings.\footnote{Thus, in our framework, although law enforcement may be biased, it is always \textit{honest}: it must disclose exculpatory evidence, and it may not misrepresent evidence or falsify evidence.}

When law enforcement is biased, we assume that it is sophisticated and is, therefore, aware that the decision-maker's belief about the defendant's guilt is lower than her own. Crucially, it believes that the decision-maker is incorrectly pessimistic, although it is actually law enforcement that is mistaken about the prior probability of guilt.\footnote{That is, when law enforcement is biased, $\mu_{DM} = \mu < \mu_{L}$; but law enforcement (mistakenly) believes $\mu_{L} = \mu > \mu_{DM}$.} Thus, when it collects evidence, law enforcement knows that it needs more evidence than should be required to exceed the conviction threshold (from her point of view) because of the decision-maker's lower prior. Similarly, when the decision-maker is biased, law enforcement is aware of this bias and knows that it needs less evidence to secure a conviction.

We find that an innocent defendant is always hurt by a biased decision-maker--he always prefers the decision-maker's bias to be as low as possible. On the other hand, an innocent defendant may prefer that law enforcement have a slight bias. A biased decision-maker affects law enforcement's information acquisition by lowering the conviction threshold--which always hurts an innocent defendant--whereas a biased law enforcement misperceives the defendant's guilt throughout her dynamic problem, which causes her either to over- or under-acquire information. It is this under-acquisition that benefits the innocent defendant. Nevertheless, despite the potential for beneficial bias when law enforcement is biased, an innocent defendant may still be better off when the decision-maker, not law enforcement, is biased. The over-investigation that arises when law enforcement is prejudiced may be worse than the lower conviction requirements insisted upon by the biased decision-maker.

Although the majority of the paper is focused on belief-based bias, in our penultimate section (Section \ref{whatkindofbias}) we explore \textit{preference-based} bias and investigate how varying the incentives, rather than the beliefs, of the agents affects an innocent defendant's welfare. We find that if it is the decision-maker who is biased, it does not matter whether his bias is preference-based or belief-based: bias effects a decrease in the $DM$'s conviction threshold, which is always bad for the defendant. In contrast, if it is law enforcement who is biased, there is a significant difference between the two types of bias. Preference-based bias shapes the incentives of law enforcement, and is similar in its effects on the defendant's welfare to bias held by the decision-maker: for several interpretations of this type of bias, the innocent defendant always prefers less bias to more bias. 

\subsection{Related Work}

In addition to the empirical papers mentioned above, there are a number of related theoretical papers. Our paper explores a static costly persuasion problem, like that explored in \cite{gentzkow2014costly}, that is equivalent to a dynamic information acquisition problem via the main result of \cite{morris2017wald}. Other related papers include \cite{fernandez2017implications}, who look at the effects of overprecision (a mistaken overconfidence in the precision of a decision-maker's information) on a decision-maker tackling a Wald problem; and \cite{augias2020wishful}, who investigate how one should optimally persuade a ``wishful thinking'' receiver. As in our paper, heterogeneous priors are important in \cite{che2009opinions}, where such differences in opinion increase the incentives for an adviser to acquire information.

Because law enforcement must disclose all of the evidence that it finds, her information acquisition is public. \cite{henry2019research} explore the problem of an agent dynamically acquiring (public) information in order to persuade a decision-maker to take a binary action. They study the effects of the decision-maker's ability to commit to decision rules and the agent's ability to misrepresent her findings. Similarly, \cite{felgenhauer2021face} investigate how an agent can persuade a decision-maker when the information acquisition problem is private and can, moreover, be altered (falsified) at a cost. Evidence production is also public in \cite{brocas2007influence}, in which an agent acquires information to sway a decision-maker with a trio of actions at his disposal. 

Both \cite{henry2009strategic} and \cite{herresthal2017hidden} compare public and private experimentation and encounter the surprising result that private experimentation may be better for the decision-maker. Conversely, \cite{janssen2018whole} also compares public versus private experimentation--where, in contrast to \cite{herresthal2017hidden}, the experimenter is informed and has state-independent preferences--and finds that private experimentation is unambiguously bad for the decision-maker. \cite{felgenhauer2014strategic} and \cite{felgenhauer2017bayesian} both look at persuasion problems in which a sender privately (and sequentially) runs experiments at a cost, the results of which she may, but need not, disclose. A striking finding of the latter paper is that the receiver prefers this private experimentation to public experimentation in all Pareto-undominated equilibria. These works suggest potentially compelling extensions of our paper: how might bias shape outcomes and the innocent defendant's welfare when exculpatory evidence need not be disclosed or when law enforcement can misrepresent or lie? Whose bias is worse for an innocent defendant then?

\section{The Model}

There is a defendant ($D$) who is being investigated for a crime. He has a private type, $\theta \in \left\{-1,1\right\}$ (innocent and guilty, respectively). The true probability that $D$ is guilty, $\mathbb{P}\left(1\right)$, is $\mu \in \left(0,1\right)$. Law enforcement ($L$) wants $D$ to be convicted, and she obtains a payoff of $v > 0$ if $D$ is convicted and a payoff normalized to $0$ if he is acquitted. 

In order to focus on the information/evidence acquisition problem faced by $L$, we simplify the conviction process greatly by assuming that there is a single decision-maker ($DM$) who convicts $D$ if and only if the $DM$'s belief exceeds an exogenously given threshold $a \in \left(\mu, 1\right]$. Both $L$ and the $DM$ have subjective priors about $D$'s guilt: $L$'s subjective belief that $D$ is guilty is $\mu_L$, and the $DM$'s subjective belief that $D$ is guilty is $\mu_{DM}$. We assume that these priors are common knowledge.\footnote{This prevents the sort of signaling through information acquisition present in \cite{jiang2021}.}

Prior to bringing $D$ to trial, $L$ may acquire information at a cost, which she must disclose. We follow \cite{morris2017wald}: $L$ samples sequentially \textit{\`{a} la} \cite{wald1945sequential} in continuous time. Each instant she observes a process $\left(Z_{t}\right)_{t \geq 0}$, where $Z_{t}$ satisfies the following stochastic differential equation:
\[dZ_{t} = \theta dt + \sigma dW_t \text{ ,}\]
where $W_t$ is Brownian motion. Each instant, $L$ may decide to keep sampling or stop acquiring information. Sampling is costly; however, and $L$ incurs a bounded and positive flow cost $c\left(\cdot\right)$ every instant that she observes $Z_{t}$. In particular, $L$ continuously updates her belief as a result of observing $Z_{t}$, and so her posterior belief follows the process $\left(\mu_{t}\right)_{t \geq 0}$. The flow cost is a function of $L$'s posterior belief: if she observes $Z_{t}$ until time $t$ she pays a total cost of $\int_{0}^{t}c\left(\mu_{s}\right)ds$. $L$'s strategy is a stopping time $\tau$. However, rather than solve this optimal-stopping problem directly, we make use of Theorem 1 in \cite{morris2017wald}, where they establish an equivalence between the dynamic Wald problem and a simpler static problem. 

Indeed, consider the static problem of $L$ choosing a distribution over posteriors $F \in \Delta \left(\left[0,1\right]\right)$ subject to a posterior-separable cost $C \colon \Delta \left(\left[0,1\right]\right) \to \mathbb{R}_{+}$. Theorem 1 of \cite{morris2017wald} states that for any flow cost $c$ in the (dynamic) Wald problem, there exists a unique function $\phi$ such that $C\left(F\right) = \int_{0}^{1}\phi(q)dF(q)$ is an equivalent static cost function, where 
\[\label{star}\tag{$\star$}\phi\left(q\right) \coloneqq \int_{\mu_{L}}^{q}\int_{\mu_{L}}^{x}\frac{\sigma^{2}c\left(y\right)}{2\left(y\left(1-y\right)\right)^2}dydx \text{ .}\]
Henceforth, we focus on the static problem, in which $L$ chooses a Bayes-plausible (\cite{kam}) distribution over posteriors, $F$,\footnote{Bayes-plausibility merely requires that $F$ has support on a subset of $\left[0,1\right]$ and that $\mathbb{E}_{F}\left[X\right] = \mu_L$.} subject to a posterior-separable cost $C\left(F\right) = \int_{0}^{1}\phi(x)dF(x)$, where $\phi$ is strictly convex. For tractability we restrict attention to cost functions $\phi$ that are twice continuously differentiable on $\left(0,1\right)$ and that are uniformly posterior separable (\cite{caplin2021rationally}).
\begin{definition}\label{assu1}
A cost function $\phi$ is \textbf{Uniformly Posterior Separable} if it is additively separable in the prior and the posterior.
\end{definition}
This class of cost functions includes the three most common cost functions used in the literature. The entropy-based cost function that has dominated the Rational Inattention literature (see e.g. \cite{sims1998stickiness}, \cite{sims2003implications}, and \cite{matvejka2015rational}),
\[\phi_{Entropy}\left(x\right) = x \log x + \left(1-x\right) \log \left(1-x\right) - \mu_L \log \mu_L - \left(1-\mu_L\right) \log \left(1-\mu_L\right) \text{ ,}\]
is one such additively separable function. As \cite{morris2017wald} show, this cost function corresponds to the dynamic problem in which the flow cost is proportional to the posterior's variance. Another cost function in this class is the log-likelihood cost function that corresponds to the dynamic problem with a constant flow cost:
\[\phi_{LL}\left(x\right) = x \log \frac{x}{1-x} + \left(1-x\right) \log \frac{1-x}{x} - \mu_L \log \frac{\mu_L}{1-\mu_L} - \left(1-\mu_L\right) \log \frac{1-\mu_L}{\mu_L} \text{ .}\]
Yet another is the Tsallis-entropy-based cost function:
\[\phi_{T}\left(x\right) = \frac{\kappa}{q-1} \left(x^q + \left(1-x\right)^q-\mu_L^q - \left(1-\mu_L\right)^q\right) \text{ ,}\]
where $q > 0$ is a scaling parameter and $\kappa > 0$ is a constant.\footnote{This form of entropy was introduced in \cite{tsallis}. \cite{bloedel2018persuasion} and \cite{caplin2021rationally} discuss this cost function in economic contexts.} $\phi_{Entropy}$ is a special case of the Tsallis cost (as $q \to 1$) and the commonly used (posterior-)variance cost function,
\[\phi_{Var}\left(x\right) = \kappa\left(x-\mu_L\right)^2 \text{ ,}\]
is the Tsallis cost when $q = 2$. This cost function corresponds to the dynamic problem in which the flow cost is proportional to the square of the posterior's variance. For some of this paper, we will further assume this particular cost function, which allows us to get a clean closed-form for $L$'s optimal learning, thereby simplifying the analysis.

\section{Bias and its Effects}

The true probability that the defendant, $D$, is guilty is $\mu \in \left(0,a\right)$. We are interested in the ramifications of a misspecified prior held by either law enforcement, $L$, or the decision-maker, the $DM$. We begin with the case in which $L$ is mistaken: $\mu_L \geq \mu = \mu_{DM}$.

\subsection{When Law Enforcement is Biased}

We suppose that the decision maker's prior belief about the defendant's likelihood of guilt is \textit{correct}; and moreover, $L$ is aware of their difference in opinions (or beliefs). Consequently, $L$ understands that given a body of evidence that induces belief $x$ from her point of view, the corresponding belief held by the $DM$, $x_{DM}\left(x\right)$ is
\[x_{DM}\left(x\right) = \frac{\mu \left(1-\mu_{L}\right)}{\mu_L\left(1-\mu\right)-x\left(\mu_{L}-\mu\right)}x \text{ .}\]
Obviously, $x_{DM}\left(x\right) \leq x$. $L$ obtains a conviction if and only if the $DM$'s posterior $x_{DM} \geq a$. Equivalently, she obtains a conviction if and only if $L$'s posterior satisfies
\[x \geq a_L \coloneqq \frac{\left(1-\mu\right)\mu_{L}}{\left(\mu_{L} - \mu\right)a + \mu\left(1-\mu_{L}\right)}a \text{ .}\]
It is easy to see that the effective conviction threshold encountered by $L$, $a_{L}$, is strictly increasing in $\mu_{L}$, which naturally implies that $a_L \geq a$, strictly so if $\mu_L > \mu$. This is the portion of the problem in which the belief heterogeneity enters $L$'s objective directly: it scales the target belief it needs to acquire to obtain a conviction.

In her \textit{perceived} static problem, $L$ chooses any distribution over posteriors supported on $[0,1]$ with mean $\mu_L$ to maximize
\[V\left(x\right) \coloneqq \begin{cases}
-\phi\left(x\right), \quad & 0 \leq x < a_L\\
v-\phi\left(x\right), \quad & a_L\leq x \leq 1\\
\end{cases} \text{ .}\]
Because of the strict convexity of $\phi$, there is a uniquely optimal distribution over posteriors $F$. It is either the degenerate distribution, $\mu_{L}$ with probability $1$; or it is a binary distribution with support on $a_L$ and some $b_L$ (with $0 \leq b_L \leq \mu_L$).\footnote{We assume that $v$ is finite and that $\phi$ is such that certainty is prohibitively costly. The entropy and the log-likelihood cost functions both satisfy this assumption. The posterior variance cost also satisfies this assumption provided $\kappa$ is sufficiently large, which we assume (see Condition \ref{conditioninterior}, below)} Solving for the optimal distribution may be done via the concavification technique of \cite{kam} but it is perhaps easier to do it via basic calculus. Indeed, $L$ solves
\[\max_{b_{L},p_{L}}\left\{- \left(1-p_{L}\right) \phi\left(b_{L}\right) + p_{L} \left(v- \phi\left(a_{L}\right) \right)\right\} \text{ ,}\]
subject to
\[0 \leq b_{L} \leq \mu_{L}, \quad 0 \leq p_{L} \leq 1, \quad \text{and} \quad p_{L} = \frac{\mu_{L} - b_{L}}{a_{L}-b_{L}} \text{ .}\]
Substituting in for $p_L$, an interior solution is pinned down by the first order condition
\[\tag{$1$}\label{eq1} 0 = \phi\left(b_{L}\right) + \left(a_{L}-b_{L}\right)\phi'\left(b_{L}\right) + v-\phi\left(a_{L}\right) \eqqcolon h\left(b_{L},\mu_{L}\right) \text{ .}\]

Another way of thinking about $L$'s problem is as choosing conditional probabilities $\gamma \coloneqq \mathbb{P}\left(g | 1\right)$ and $\lambda \coloneqq \mathbb{P}\left(g | -1\right)$. Using Bayes' law, we can write these in terms of $a_{L}$ and $p_{L}$:
\[\label{eq2}\tag{$2$}\gamma = \frac{a_{L}}{\mu_{L}} p_{L}, \quad \text{and} \quad \lambda = \frac{1-a_{L}}{1-\mu_{L}} p_{L} \text{ .}\]
Combining the Equations in \ref{eq2}, we see that 
\[\gamma=\lambda \frac{1-\mu_L}{\mu_L}\frac{a_L}{1-a_L}=\lambda\frac{1-\mu}{\mu}\frac{a}{1-a}>\lambda \text{ .}\]
That is, $\gamma$ is a fixed multiple of $\lambda$, and is; moreover, independent of $\mu_L$. Since the true unconditional probability that a defendant is convicted is $p \coloneqq \mu \gamma + \left(1-\mu\right) \lambda$, we have
\begin{remark}\label{signingbias}
The true unconditional probability of conviction, $p$, is (strictly) increasing in $L$'s prior, $\mu_{L}$; if and only if the true probability of convicting an innocent $D$, $\lambda$, is (strictly) increasing in $\mu_{L}$; if and only if the true probability of convicting a guilty $D$, $\gamma$, is (strictly) increasing in $\mu_{L}$.
\end{remark}
Equivalently, the effect of the false prior on the false positive rate must always be in the same direction as the effect of the false prior on the true positive rate. Thus, a bias level that increases the occurrence of type I errors reduces the rate of type II errors and vice-versa.

In order to get more concrete results, we now specify a flow cost function that is proportional to the square of the posterior variance which yields the variance cost, $\phi_{Var} = \kappa \left(x-\mu_{L}\right)^2$, in $L$'s static problem. We also impose the following condition on the parameters, which guarantees that $L$ never becomes certain that $D$ is innocent.
\begin{condition}\label{conditioninterior}
$a_{DM} > \sqrt{v/\kappa} \eqqcolon d$, where $a_{DM}$ is the cutoff belief $L$ needs to obtain to secure a conviction when the $DM$ is biased.\footnote{We derive $a_{DM} = \frac{\left(1-\mu_{DM}\right)\mu}{\left(\mu - \mu_{DM}\right)a + \mu_{DM}\left(1-\mu\right)}a$, below, when we look at the situation with a biased decision-maker.}
\end{condition}
Furthermore, note that by construction $a_L \geq a \geq a_{DM}$. Intuitively, Condition \ref{conditioninterior} insists that the ratio between reward and information cost cannot be too high.

Substituting the variance cost function into Equation \ref{eq1}, we see that the interior solution is for $L$ to choose a binary distribution with support $\left\{a_L-d, a_L\right\}$, which is uniquely optimal provided $\mu_L \in \left(a_L-d, a_L\right)$. Phrased in terms of the dynamic problem, in the interior solution, $L$ acquires evidence until she reaches belief $x_{h} = a_L$, at which point she takes the evidence to $DM$ to obtain a conviction; or until she reaches belief $x_{l} = a_L - d$, at which point she stops looking for evidence, being too pessimistic about $D$'s guilt to continue. If $\mu_L \leq a_L-d$, $L$ does not acquire any evidence since she believes the odds are overly stacked against the defendant being guilty.

Substituting in for $a_{L}$ and rearranging $\mu_L > a_L - d$, we see that $L$ acquires evidence provided \[d > \frac{\left(a-\mu\right)\left(1-\mu_L\right)\mu_L}{\left(\mu_L-\mu\right)a+\left(1-\mu_L\right)\mu} \eqqcolon \ubar{d} \text{ .}\]

Our next result is somewhat counterintuitive: we show that there are parameters for which a biased $L$ does not acquire evidence (which is to $D$'s benefit since he is therefore convicted with probability $0$) whereas an unbiased $L$ does.

\begin{lemma}\label{when2investigate}
Assume the variance cost, $\phi = \phi_{Var}$. If the prior probability of guilt and the conviction are sufficiently low, $1 > \mu + a$, the minimal reward ratio required for investigation, $\ubar{d}$, is strictly increasing in $L$'s prior, $\mu_L$, for all $\mu_L$ sufficiently close to $\mu$. Otherwise (if $1 < \mu + a$), $\ubar{d}$ is strictly decreasing in $\mu_L$.
\end{lemma}
\textit{Viz.}, a biased $L$ may actually investigate (and therefore convict) defendants less frequently than a non-biased $L$. Active information acquisition requires that $d > a_L - \mu_L$ and so $\ubar{d}$ is strictly increasing in $\mu_L$ if and only if $a_L - \mu_L$ is strictly increasing in $\mu_L$. Moreover, $p_L = \left(\mu_L + d-a_L\right)/d$, and so it is easy to see that $\ubar{d}$ is strictly increasing in $\mu_L$ if and only if $p_L$ is strictly decreasing in $\mu_L$. This is intuitive: crossing the threshold from active information acquisition to no information acquisition requires that $L$ shrink from positive to negative, the latter of which is obviously impossible, requiring a corner solution (no information acquisition).

The key term in the expression for $p_L$ is $\mu_L - a_L$. $a_L$ is strictly increasing in $\mu_L$ and so this second term allows $\mu_L$ to exert a negative influence on $p_L$. The first term, obviously, is where $\mu_L$ has a positive influence. It is when the second term (and $\mu$ and $\mu_L$) is (are) small that $a_L$ reacts strongly to a change in $\mu_L$, yielding the result. That $\mu_L$ cannot be too extreme is a common theme throughout our results: a severe bias is bad for $D$, regardless of who is biased. 

Turning our attention to the case when there is information acquisition, we may use the equations in \ref{eq2} to derive explicit expressions for $\gamma$ and $\lambda$, which are
\[\gamma = \left(\frac{a_L}{\mu_L}\right)\frac{\mu_L + d-a_L}{d}, \quad \text{and} \quad \lambda = \left(\frac{1-a_L}{1-\mu_L}\right)\frac{\mu_L + d-a_L}{d} \text{ .}\]
In the same spirit as Lemma \ref{when2investigate}, the next theorem details when $D$ benefits from $L$'s bias.

\begin{theorem}\label{mainthm1}
Assume the variance cost, $\phi = \phi_{Var}$. 
\begin{enumerate}[label=\roman*]
    \item Let the correct prior, reward ratio, and conviction threshold ($\mu$, $d$, and $a$, respectively) be such that if she were not biased $L$ would acquire information ($\mu > a - d$). If $a$ is sufficiently low and $d$ is sufficiently high ($2 a - d < 1$), then for all $\mu_L$ sufficiently close to $\mu$, an innocent $D$ (strictly) prefers for $L$ to be biased rather than unbiased. Otherwise, an innocent $D$ prefers for $L$ to be unbiased rather than biased, for any amount of bias.
    \item If $\mu$, $d$, and $a$ are such that if she were not biased $L$ would \textit{not} acquire information ($\mu \leq a - d$), an innocent $D$ (weakly) prefers for $L$ to have no bias.
\end{enumerate}
\end{theorem}
In fact, an even stronger result holds: if $a$ is sufficiently low and $d$ is sufficiently high (and $\mu > a - d$), then for all $\mu_L$ sufficiently close to $\mu$, $\lambda$, the probability that an innocent $D$ is convicted is (strictly) decreasing in $\mu_L$. Not only is bias better for $D$ than no bias, but $D$'s welfare is actually increasing in the amount of bias. Outside of this interval, $D$'s welfare is decreasing in the amount of bias.

The last statement of Theorem \ref{mainthm1} is trivial: if the parameters are such that $L$ does not investigate $D$, $D$ can only be (weakly) hurt by any bias. In addition, the statement of this theorem is phrased in terms of the outcome for the innocent $D$, but recall that bias has the same affect on outcome for the guilty $D$. Accordingly, $D$ prefers $L$ have a mild bias provided the bar for conviction is not too high and $L$'s reward for a conviction is not too low.

Figure \ref{fig4} illustrates two specifications of the parameters. The first, Case 1, is an instance in which bias benefits $D$. The bottom function is $L$'s perceived value function, given her incorrect prior. She acquires evidence until her posterior is $a_L$ or until her posterior is $b_L$. Importantly, $x_{DM}\left(b_L\right)$ is the correct belief that $L$ should have after acquiring the evidence that yields her belief $b_L$. It is always less than $b_L$, but (in general) it may or may not be less than $b$, the low posterior in the counterfactual problem in which $L$ is not biased. It is precisely when $x_{DM}\left(b_L\right) > b$ that bias strictly benefits $D$: $L$ mistakenly stops her information acquisition too early, earlier than she would have, had she the correct prior. Accordingly, Case 2, is an instance in which bias strictly hurts $D$. There, $x_{DM}\left(b_L\right) < b$: $L$ over-investigates, making it more likely that she encounters the evidence she needs to convict $D$, to $D$'s detriment.

\begin{figure}[tbp]
\centering
\begin{subfigure}{.5\textwidth}
  \centering
  \includegraphics[scale=.15]{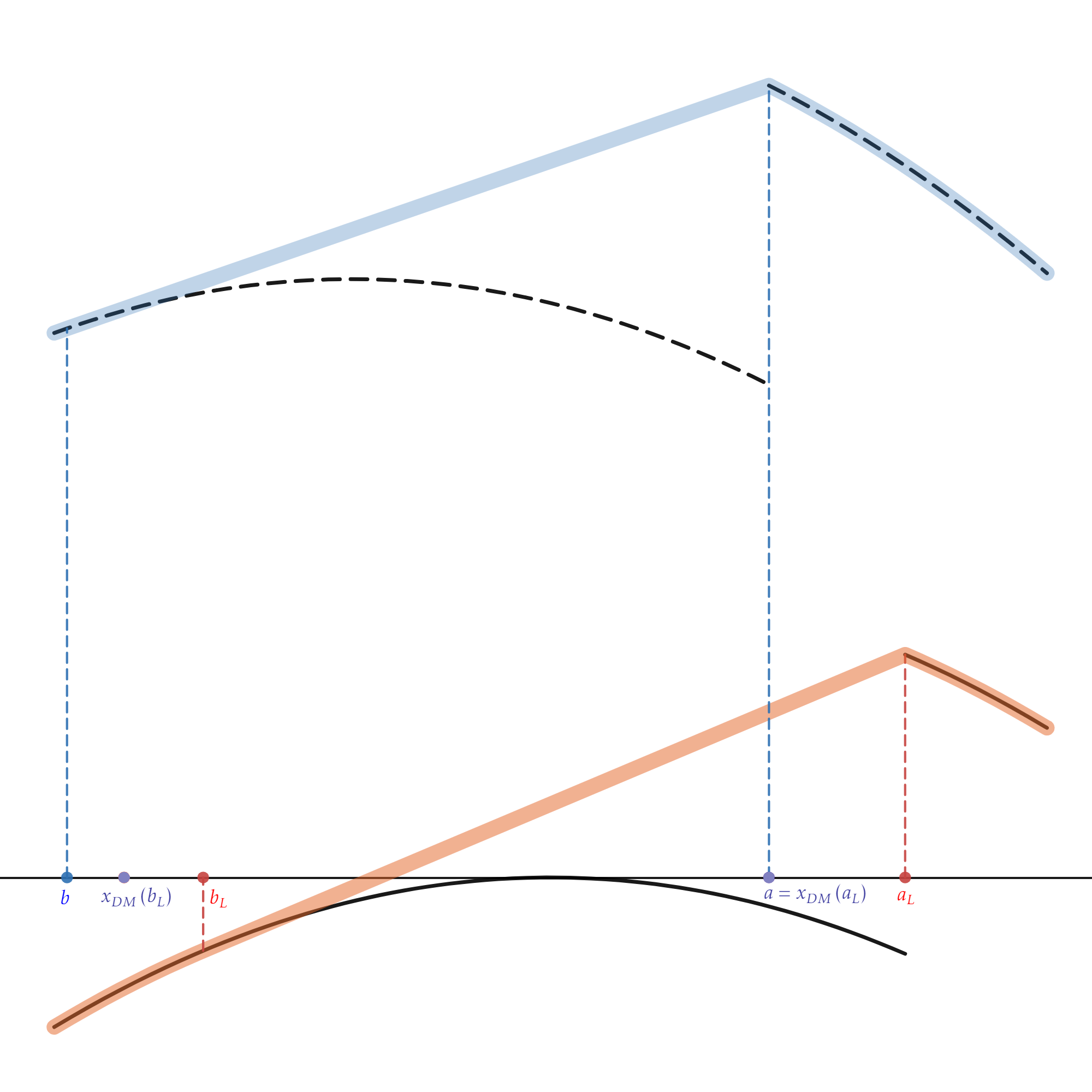}
  \caption{\textbf{Case 1:} Low $a$.}
  \label{figsub12}
\end{subfigure}%
\begin{subfigure}{.5\textwidth}
  \centering
  \includegraphics[scale=.15]{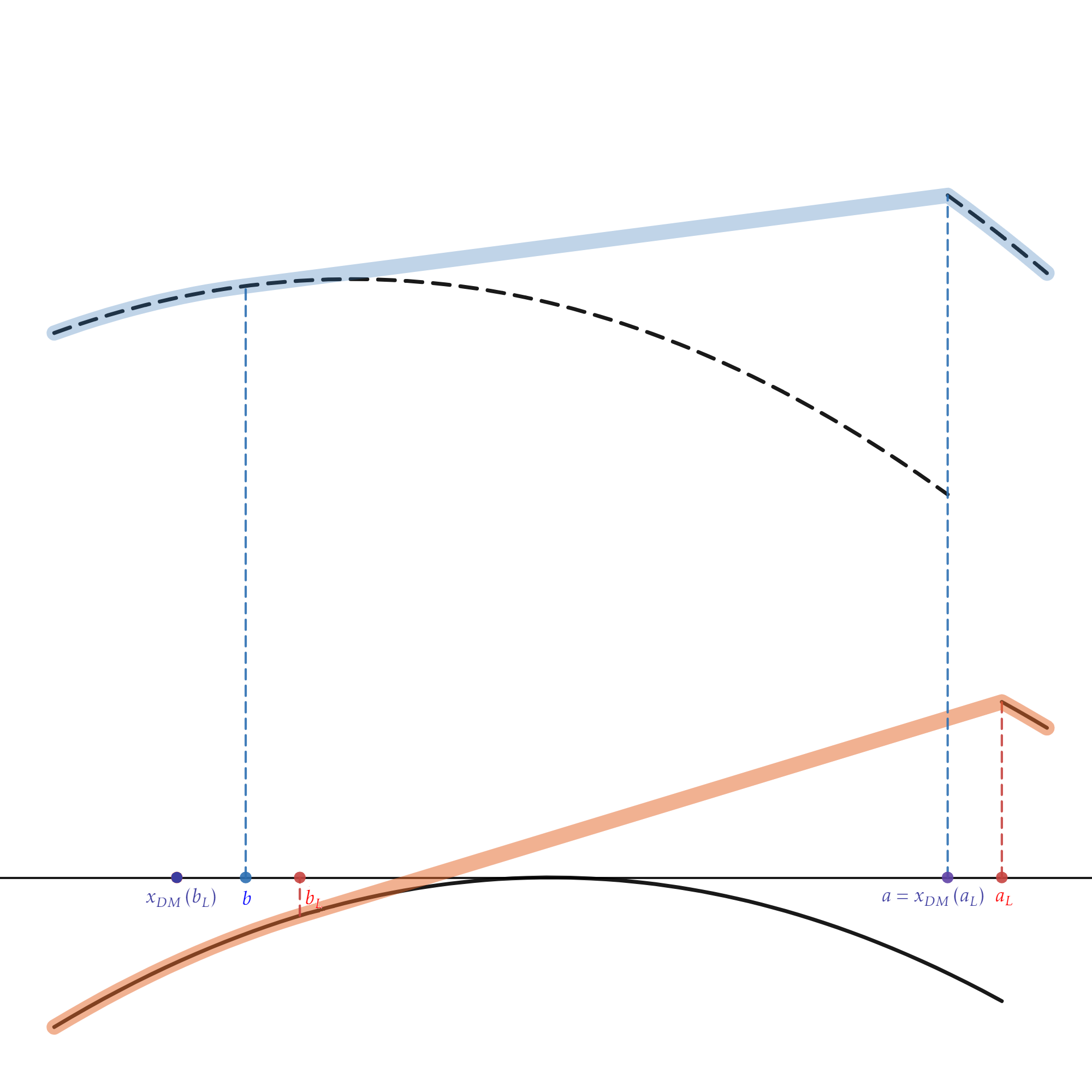}
  \caption{\textbf{Case 2:} High $a$.}
  \label{figsub22}
\end{subfigure}
\caption{$L$'s perceived value function, $V\left(x\right)$, (solid black); its concavification (red); $L$'s actual value function (dashed black); and its concavification (blue).}
\label{fig4}
\end{figure}

\subsection{When the Decision-Maker is Biased}

We now stipulate that it is the decision-maker who is biased against the defendant and has a prior belief $\mu_{DM} > \mu = \mu_{L}$ that $D$ is guilty. If $L$ induces $x$, the corresponding belief held by the $DM$ is
\[x_{DM}\left(x\right) = \frac{\mu_{DM} \left(1-\mu\right)}{\mu\left(1-\mu_{DM}\right)-x\left(\mu-\mu_{DM}\right)}x \text{ .}\]
Thus, $L$ secures a conviction if and only if it induces a belief
\[x \geq a_{DM} \coloneqq \frac{\left(1-\mu_{DM}\right)\mu}{\left(\mu - \mu_{DM}\right)a + \mu_{DM}\left(1-\mu\right)}a \text{ .}\]

Using Bayes' law, we can write out the conditional probabilities
$\hat{\gamma} \coloneqq \mathbb{P}\left(g | 1\right)$ and $\hat{\lambda} \coloneqq \mathbb{P}\left(g | -1\right)$:
\[\hat{\gamma} = \frac{a_{DM}}{\mu} p_{DM}, \quad \text{and} \quad \hat{\lambda} = \frac{1-a_{DM}}{1-\mu} p_{DM} \text{ ,}\]
where $p_{DM}$ is pinned down by the first order condition
\[\label{eq2n2}\tag{$2$}0 = \phi\left(b_{DM}\right) + \left(a_{DM}-b_{DM}\right)\phi'\left(b_{DM}\right) + v-\phi\left(a_{DM}\right) \eqqcolon f\left(b_{DM},a_{DM}\right) \text{ ,}\]
and $p_{DM} = \left(\mu-b_{DM}\right)/\left(a_{DM}-b_{DM}\right)$.\footnote{Keep in mind that $b_{DM}$ is $L$'s (correct) low belief in her optimal learning, where the subscript ``$DM$'' indicates that this is the case with the biased $DM$. Likewise, $p_{DM}$ is the (correct) probability of conviction.}

The conviction threshold $a_{DM}$ faced by $L$ is strictly decreasing in the $DM$'s prior, $\mu_{DM}$. If this threshold, $a_{DM}$, is weakly less than $\mu$, (equivalently, if $a \leq \mu_{DM}$) law enforcement does not obtain any information and secures a conviction with probability $1$. Moreover, as we establish in proof of Proposition \ref{biasinc}, the lower bound for $L$'s learning, $b_{DM}$, is also strictly decreasing in $\mu_{DM}$. Consequently, any increase in the $DM$'s bias that either renders information acquisition optimal or sub-optimal can only hurt the innocent $D$. As it turns out, we are able to prove a general result concerning the effect of prejudice on an innocent $D$'s welfare.
\begin{proposition}\label{biasinc}
The probability that an innocent $D$ is convicted, $\hat{\lambda}$, is increasing in the $DM$'s bias, strictly so if $L$ would acquire information in the problem without bias.
\end{proposition}
This result is intuitive, once we make the observation that an increase in the $DM$'s bias is identical from $L$'s (and $D$'s) perspective to a decrease in the threshold needed to secure conviction. A lower conviction threshold obviously hurts $D$ overall, and this proposition establishes that it hurts an innocent $D$ to boot.

\subsection{Whose Bias?}

So far, we have shown that both an innocent and a guilty $D$ may benefit from a biased $L$, but that an innocent $D$ is always worse off when the $DM$ is biased. Suppose we fix the level of bias. Does an innocent $D$ always prefer that $L$ is biased rather than the $DM$? For convenience, we set $\mu_B \geq \mu$ to be the biased belief--we shall now compare the case in which $\mu_B = \mu_{DM}$ with the case in which $\mu_B = \mu_L$.

It is easy to come up with conditions that ensure an innocent $D$ would rather $L$ have the bias. For instance, if the conditions in Theorem \ref{mainthm1} hold--that guarantee that $L$'s bias benefits $D$--then a biased $L$ is better for $D$ than a biased $DM$ (since we show in Proposition \ref{biasinc} that an innocent $D$'s welfare is decreasing in the $DM$'s bias). Furthermore, if the bias is sufficiently large ($\mu_{B} \geq a$), then a biased $DM$ requires no convincing and will convict based on the prior alone. Given this extreme prejudice, $L$ will not acquire any information, and an innocent $D$ will be convicted with probability $1$. Conversely, a biased $L$ never secures a conviction with probability $1$, no matter her bias. These observations, plus the the continuity of  $\hat{\lambda}$ and $\lambda$ in $\mu_B$, allow us to conclude the following remark.
\begin{remark}
For any biased prior, $\mu_B$, that is sufficiently large, an innocent $D$'s welfare when $L$ is biased is strictly higher than an innocent $D$'s welfare when the $DM$ is biased.
\end{remark}
Obviously, for all $\mu_B$ sufficiently large, a guilty $D$ is also better off when $L$, not the $DM$, is prejudiced. 

These observations notwithstanding, an innocent $D$ may still prefer that the biased party be the $DM$ and not $L$. We begin with this result for the variance cost functional, $C\left(F\right) = \int \phi_{Var}dF$.

\begin{proposition}\label{lastproposition}
Assume the variance cost: $\phi = \phi_{Var}$. Let the correct prior, reward ratio, and conviction threshold ($\mu$, $d$, and $a$, respectively) be such that if she were not biased $L$ would acquire information ($\mu > a - d$). If $a$ is sufficiently high and if $d$ is sufficiently small, there exists a threshold (false) prior $\mu_{B}^{\diamond}$, such that an innocent $D$ prefers the $DM$ to be biased to $L$ if and only if $\mu_{B} \in \left(\mu, \mu_{B}^{\diamond}\right)$. Otherwise, an innocent $D$ prefers $L$ to be biased for all $\mu_{B} \geq \mu$.
\end{proposition}

It is illuminating to contrast this result to Theorem \ref{mainthm1}, in which we establish that a large reward (high $d$) and a low conviction threshold ($a$) are required in order for bias to benefit an innocent $D$. Conversely, Proposition \ref{lastproposition} states that a \textit{low} reward and a \textit{high} conviction threshold are essential parts of an environment in which $D$ prefers the $DM$'s prejudice to $L$'s. Regardless of who is biased, $D$ is hurt by bias if and only if the bias causes $L$ to over-investigate (relative to the situation when there is no bias).\footnote{As an aside, one might think that an innocent defendant should prefer more investigation to less, so that the truth is more likely to be discovered. However, this is only true if the current body of evidence is such that the defendant would be convicted. An adversarial investigator would never do such a thing, rendering an increase in scrutiny always bad in our scenario.} When the $DM$ is prejudiced, this excess of information acquisition is guaranteed and is due to the effective decrease in the conviction threshold. The cause of $L$'s over-investigation due to her own bias is more subtle: it occurs when the (false) conviction probability $L$ is increasing sufficiently quickly in $\mu_B$. An even greater rate of increase, thus, is needed for $L$'s partisanship to be worse for $D$ than the $DM$'s: it is then that the over-investigation $L$ is ``tricked'' into by her false prior exceeds the over-investigation engendered by the conviction threshold's decrease.

Can the findings of Proposition \ref{lastproposition} be generalized? Note that the conditions in this proposition ensure that the bias is not too large and that the likelihood of conviction is small; \textit{viz.},, both the perceived likelihood, $p_L$, from the scenario in which $L$ is biased; and the true likelihood, $p_{DM}$, from the scenario in which the $DM$ is biased--which are close to one another when the bias is small--are small. Accordingly, the effect of a small amount of bias is approximated well by rate of change of the likelihood of conviction (again, $p_{DM}$ versus $p_L$) in the size of the bias. 

\begin{theorem}\label{bigtheorem}
Suppose one of the following three conditions holds:
\begin{enumerate}[label=\roman*]
    \item The cost is the log-likelihood cost, $\phi = \phi_{LL}$, and the correct prior is low, $\mu < 1/2$;
    \item The cost is the Tsallis cost, $\phi = \phi_{T}$ (which includes $\phi = \phi_{Entropy}$); or
    \item The cost, $\phi$, has bounded slope on $[0,1]$.
\end{enumerate}
Let the conviction threshold $a = 1 - \epsilon$, $\epsilon > 0$, and define $v_\epsilon$ to be the lowest reward ratio, $v$, such that there is an interior solution when there is zero bias, $\mu_B = \mu$. Then, for all $\epsilon$ sufficiently small, there exists an interval $\left(v_\epsilon,\bar{v}\right)$ such that an innocent $D$'s welfare is strictly higher when the $DM$ is biased instead of $L$ for all $\mu_B$ sufficiently close to $\mu$.
\end{theorem}

In short, we find if the conviction threshold is sufficiently high, and $L$ is very close to moving on from the suspect (abandoning her investigation of $D$), \textit{and} bias is not too large, the $DM$'s bias is better for an innocent defendant than $L$'s. The over-investigation that arises due to $L$'s bias is worse for the innocent defendant than the lower conviction threshold induced by the $DM$'s bias.

\section{What Kind of Bias?}\label{whatkindofbias}

So far, we have modeled bias as an incorrect prior. That is, we characterize bias in terms of agents' beliefs. However, bias may also manifest through agents' preferences. We can model this preference-based bias as follows: suppose now that law enforcement's payoff from conviction
is of the form \[\eta u\left(\alpha,\theta\right) + \left(1-\eta\right) v \mathbbm{1}_{\alpha = convict} \text{ ,}\]
where $\alpha \in \left\{acquit, convict\right\}$ is the decision-maker's action and where
\[u\left(\alpha,\theta\right) = \begin{cases}
0, \quad &\text{if} \ \left(\alpha,\theta\right) = \left(convict, 1\right) \ \text{or} \ \left(\alpha,\theta\right) = \left(acquit, -1\right)\\
-1, \quad &\text{if} \ \left(\alpha,\theta\right) = \left(acquit, 1\right)\\
-\rho, \quad &\text{if} \ \left(\alpha,\theta\right) = \left(convict, -1\right)
\end{cases} \text{ ,}\]
where $\rho \geq 0$. We impose that the $DM$'s conviction threshold $a \geq \frac{\rho\eta-\left(1-\eta\right)v}{\eta\left(1+\rho\right)}$, so that the $DM$ has a higher conviction threshold than $L$. When $\eta = 0$, we are in the scenario of the earlier portion of the paper: in terms of preferences, $L$ is maximally biased in our main specification. In addition, both $DM$ and $L$ share the same (correct) prior $\mu$ that $D$ is guilty.

$L$ chooses any Bayes-plausible distribution to maximize
\[V\left(x\right) \coloneqq \begin{cases}
-\eta x -\phi\left(x\right), \quad & 0 \leq x \leq a\\
-\eta \rho \left(1-x\right) + \left(1-\eta\right)v-\phi\left(x\right), \quad & a \leq x \leq 1\\
\end{cases} \text{ .}\]
It is easy to see that there are two varieties of interior solution: either $L$ obtains beliefs $a$ and $b \ (< \mu)$ (the $DM$'s preferences bind) or $L$ obtains beliefs $x_H > a$ and $b \ (< \mu)$ (the $DM$'s preferences do not bind). If the $DM$'s preferences bind, $L$'s behavior (at an interior solution) is pinned down by \[\tag{$\triangle$}\label{triangle}0 = \phi\left(b\right) + \eta a + \left(a-b\right)\phi'\left(b\right) - \eta \rho \left(1-a\right)  + \left(1-\eta\right)v-\phi\left(a\right) \eqqcolon s\left(b,a\right) \text{ .}\]
Henceforth, we make the following assumption, which is necessary and sufficient for $L$'s solution to be interior, in which the $DM$'s preferences bind.
\begin{assumption}\label{ass2}
$b > 0$, $\mu \in \left(b,a\right)$, and $-\eta - \phi'\left(b\right) \geq \eta \rho - \phi'\left(a\right) $, where $b$ is the unique root of $s$ defined in Expression \ref{triangle}.
\end{assumption}

The first result of this section states that the $DM$'s bias has the same effect on an innocent $D$'s welfare as in the main specification. Importantly, preference-based and belief-based bias for the $DM$ have the same effect on outcomes: both lower the $DM$'s conviction threshold.
\begin{proposition}\label{firstpreference}
The probability that an innocent $D$ is convicted is increasing in the $DM$'s bias.
\end{proposition}
There are a number of ways of modeling $L$'s preference-based bias. An increase in this bias could correspond to an increase in $v$, or a decrease in $\eta$ or $\rho$. Given this multiplicity, we establish the following result concerning the effect of $L$'s preference-based bias.
\begin{proposition}\label{finalpropprop}
When an increase in bias corresponds to an increase in $L$'s payoff from conviction, $v$, or a decrease in penalty from convicting an innocent defendant, $\rho$, the probability that an innocent $D$ is convicted is increasing in $L$'s bias. 

When an increase in bias corresponds to a decrease in the weight the $DM$ places on getting the decision correct, $\eta$, the probability that an innocent $D$ is convicted is increasing in $L$'s bias if and only if $v$ and $\rho$ are sufficiently large and the conviction threshold, $a$, is sufficiently small ($v \geq a - \rho \left(1-a\right)$).
\end{proposition}
It is interesting to contrast this proposition with Theorem \ref{mainthm1}. Here, interpreting a smaller $\eta$ as an increase in $L$'s (preference-based) bias, it is only when $a$ is sufficiently high and $v$ is sufficiently low that an innocent $D$ benefits from bias. Moreover, an innocent $D$ never benefits from bias, when a larger $v$ or a smaller $\rho$ corresponds to a more biasd $L$. But we should not be surprised that preference-based bias has such a different effect to belief-based bias. With preference-based bias, it does not matter so much who ($L$ or the $DM$) is the biased party since in both cases bias works to shift the incentives of $L$. In contrast, if $L$ suffers from belief-based bias, there is no change in her incentives, since she is aware of the $DM$'s lack of bias and tailors the high stopping threshold accordingly. Instead, her stopping problem is altered, unbeknownst to her. Accordingly, it is not surprising that under two of the three interpretations of (preference-based) bias, Propositions \ref{firstpreference} and \ref{finalpropprop} reveal that it does not matter who has the bias: an innocent $D$ is always worse off.

One interesting parallel between Proposition \ref{finalpropprop} and the earlier Proposition \ref{signingbias} is that independent of whether bias is preference- or belief-based, the guilty $D$ is affected in the same manner as the innocent $D$ and hence the unconditional effect of bias also has the same sign. That is, if an innocent $D$ is made worse off by bias (regardless of whether it is preference- or belief-based), so is a guilty $D$. Accordingly, a change in the level of bias effects an exchange of type I errors for type II errors, or vice-versa.

\section{Discussion}

In this paper, we take an agent's bias as given, and explore its ramifications. Perhaps unexpectedly, our results challenge the notion that bias must necessarily be bad for a suspect: a prejudiced law-enforcement may be to suspect's benefit, due to its effect on the evidence-gathering process. On the other hand, a partisan decision-maker (judge, jury) is always worse for a defendant than a fair one. Given these observations, one might suppose that for a fixed level of bias, a defendant should always prefer that law-enforcement be biased, but our findings contest that result as well.

Under a max-min criterion, this paper has clear policy implications: efforts in eliminating prejudice should be applied to decision-makers (judges and juries), since such efforts are guaranteed to improve equity. Whether a reduction in law-enforcement's bias does so is more nebulous.\footnote{Of course, there are other benefits from having a less prejudiced law-enforcement that our model does not capture. For instance, police interactions occasionally turn violent (even deadly), and our framework does not address the consequences of bias along this dimension.}

Furthermore, although the motivation for this paper was the specific legal environment that we explore, our findings apply to other settings in which there is delegated (and motivated) costly information acquisition. For instance, we can replace law enforcement with a pharmaceutical company who wants its drug approved by a regulator ($DM$), and who runs costly trials to obtain evidence about the drug's effectiveness. Our framework allows us to explore how the approval rate of ineffective (or dangerous drugs)--corresponding to the low state $\theta = -1$--is affected by the subjective beliefs of the two parties. Similarly, a headhunter in a CEO search can take the role of law enforcement, and the firm looking for a CEO would be the decision-maker. In this scenario, our results shed light on the effects of discrimination both belief- and, thanks to Section \ref{whatkindofbias}, taste-based.

One final comment: the perspicacious reader may have picked up on the fact that \textit{the true prior, $\mu$, is irrelevant} in our model. Although we specify that the prior of one of the actors is correct and that the other is biased (has a higher prior) this is just to ease presentation and discussion. The conditional probabilities that innocent and guilty defendants are convicted depend \textit{only} on the subjective beliefs held by the decision-maker and law enforcement. Naturally, the unconditional probability that a defendant is convicted is affected by the true likelihood of guilt.

\bibliography{sample.bib}

\begin{thebibliography}{28}
\providecommand{\natexlab}[1]{#1}
\providecommand{\url}[1]{\texttt{#1}}
\expandafter\ifx\csname urlstyle\endcsname\relax
  \providecommand{\doi}[1]{doi: #1}\else
  \providecommand{\doi}{doi: \begingroup \urlstyle{rm}\Url}\fi

\bibitem[Abrams et~al.(2012)Abrams, Bertrand, and
  Mullainathan]{abrams2012judges}
David~S Abrams, Marianne Bertrand, and Sendhil Mullainathan.
\newblock Do judges vary in their treatment of race?
\newblock \emph{The Journal of Legal Studies}, 41\penalty0 (2):\penalty0
  347--383, 2012.

\bibitem[Anwar et~al.(2012)Anwar, Bayer, and Hjalmarsson]{anwar2012impact}
Shamena Anwar, Patrick Bayer, and Randi Hjalmarsson.
\newblock The impact of jury race in criminal trials.
\newblock \emph{The Quarterly Journal of Economics}, 127\penalty0 (2):\penalty0
  1017--1055, 2012.

\bibitem[Augias and Barreto(2020)]{augias2020wishful}
Victor Augias and Daniel Barreto.
\newblock Wishful thinking: Persuasion and polarization.
\newblock \emph{arXiv preprint arXiv:2011.13846}, 2020.

\bibitem[Bloedel and Segal(2020)]{bloedel2018persuasion}
Alexander~W Bloedel and Ilya Segal.
\newblock Persuading a rationally inattentive agent.
\newblock \emph{Working paper}, 2020.

\bibitem[Brocas and Carrillo(2007)]{brocas2007influence}
Isabelle Brocas and Juan~D Carrillo.
\newblock Influence through ignorance.
\newblock \emph{The RAND Journal of Economics}, 38\penalty0 (4):\penalty0
  931--947, 2007.

\bibitem[Caplin et~al.(2021)Caplin, Dean, and Leahy]{caplin2021rationally}
Andrew Caplin, Mark Dean, and John Leahy.
\newblock Rationally inattentive behavior: Characterizing and generalizing
  shannon entropy.
\newblock \emph{Forthcoming, Journal of Political Economy}, 2021.

\bibitem[Che and Kartik(2009)]{che2009opinions}
Yeon-Koo Che and Navin Kartik.
\newblock Opinions as incentives.
\newblock \emph{Journal of Political Economy}, 117\penalty0 (5):\penalty0
  815--860, 2009.

\bibitem[Cohen and Yang(2019)]{cohen2019judicial}
Alma Cohen and Crystal~S Yang.
\newblock Judicial politics and sentencing decisions.
\newblock \emph{American Economic Journal: Economic Policy}, 11\penalty0
  (1):\penalty0 160--91, 2019.

\bibitem[Devi and Fryer~Jr(2020)]{devi2020policing}
Tanaya Devi and Roland~G Fryer~Jr.
\newblock Policing the police: The impact of ``pattern-or-practice''
  investigations on crime.
\newblock \emph{Mimeo}, 2020.

\bibitem[Felgenhauer and Loerke(2017)]{felgenhauer2017bayesian}
Mike Felgenhauer and Petra Loerke.
\newblock Bayesian persuasion with private experimentation.
\newblock \emph{International Economic Review}, 58\penalty0 (3):\penalty0
  829--856, 2017.

\bibitem[Felgenhauer and Schulte(2014)]{felgenhauer2014strategic}
Mike Felgenhauer and Elisabeth Schulte.
\newblock Strategic private experimentation.
\newblock \emph{American Economic Journal: Microeconomics}, 6\penalty0
  (4):\penalty0 74--105, 2014.

\bibitem[Felgenhauer and Xu(2021)]{felgenhauer2021face}
Mike Felgenhauer and Fangya Xu.
\newblock The face value of arguments with and without manipulation.
\newblock \emph{International Economic Review}, 62\penalty0 (1):\penalty0
  277--293, 2021.

\bibitem[Fernandez and Mayskaya(2017)]{fernandez2017implications}
Marcelo Fernandez and Tatiana Mayskaya.
\newblock Implications of overconfidence on information investment.
\newblock \emph{Available at SSRN 3038356}, 2017.

\bibitem[Fryer~Jr(2019)]{fryer2019empirical}
Roland~G Fryer~Jr.
\newblock An empirical analysis of racial differences in police use of force.
\newblock \emph{Journal of Political Economy}, 127\penalty0 (3):\penalty0
  1210--1261, 2019.

\bibitem[Gentzkow and Kamenica(2014)]{gentzkow2014costly}
Matthew Gentzkow and Emir Kamenica.
\newblock Costly persuasion.
\newblock \emph{American Economic Review}, 104\penalty0 (5):\penalty0 457--62,
  2014.

\bibitem[Henry(2009)]{henry2009strategic}
Emeric Henry.
\newblock Strategic disclosure of research results: The cost of proving your
  honesty.
\newblock \emph{The Economic Journal}, 119\penalty0 (539):\penalty0 1036--1064,
  2009.

\bibitem[Henry and Ottaviani(2019)]{henry2019research}
Emeric Henry and Marco Ottaviani.
\newblock Research and the approval process: the organization of persuasion.
\newblock \emph{American Economic Review}, 109\penalty0 (3):\penalty0 911--55,
  2019.

\bibitem[Herresthal(2017)]{herresthal2017hidden}
Claudia Herresthal.
\newblock Hidden testing and selective disclosure of evidence.
\newblock \emph{Mimeo}, 2017.

\bibitem[Janssen(2018)]{janssen2018whole}
Mathijs Janssen.
\newblock The whole truth? generating and suppressing hard evidence.
\newblock \emph{Mimeo}, 2018.

\bibitem[Jiang(2021)]{jiang2021}
Shaofei Jiang.
\newblock Costly persuasion by a partially informed sender.
\newblock \emph{Mimeo}, 2021.

\bibitem[Kamenica and Gentzkow(2011)]{kam}
Emir Kamenica and Matthew Gentzkow.
\newblock Bayesian persuasion.
\newblock \emph{The American Economic Review}, 101\penalty0 (6):\penalty0
  2590--2615, 2011.

\bibitem[Mat{\v{e}}jka and McKay(2015)]{matvejka2015rational}
Filip Mat{\v{e}}jka and Alisdair McKay.
\newblock Rational inattention to discrete choices: A new foundation for the
  multinomial logit model.
\newblock \emph{American Economic Review}, 105\penalty0 (1):\penalty0 272--98,
  2015.

\bibitem[Morris and Strack(2017)]{morris2017wald}
Stephen Morris and Philipp Strack.
\newblock The wald problem and the equivalence of sequential sampling and
  static information costs.
\newblock \emph{Mimeo}, 2017.

\bibitem[Shayo and Zussman(2011)]{shayo2011judicial}
Moses Shayo and Asaf Zussman.
\newblock Judicial ingroup bias in the shadow of terrorism.
\newblock \emph{The Quarterly Journal of Economics}, 126\penalty0 (3):\penalty0
  1447--1484, 2011.

\bibitem[Sims(1998)]{sims1998stickiness}
Christopher~A Sims.
\newblock Stickiness.
\newblock In \emph{Carnegie-rochester conference series on public policy},
  volume~49, pages 317--356, 1998.

\bibitem[Sims(2003)]{sims2003implications}
Christopher~A Sims.
\newblock Implications of rational inattention.
\newblock \emph{Journal of Monetary Economics}, 50\penalty0 (3):\penalty0
  665--690, 2003.

\bibitem[Tsallis(1988)]{tsallis}
Constantino Tsallis.
\newblock Possible generalization of boltzmann-gibbs statistics.
\newblock \emph{Journal of Statistical Physics}, 1988.

\bibitem[Wald(1945)]{wald1945sequential}
Abraham Wald.
\newblock Sequential tests of statistical hypotheses.
\newblock \emph{The Annals of Mathematical Statistics}, 16\penalty0
  (2):\penalty0 117--186, 1945.

\end{thebibliography}

\appendix

\section{Omitted Proofs}

\subsection{Lemma \ref{when2investigate} Proof}
\begin{proof}
A direct calculation reveals that $\ubar{d}'\left(\mu_{L}\right)$ has the same sign as \[\left(1-a\right)\mu\left(1-2\mu_{L}\right) - \left(a-\mu\right)\mu_{L}^2 \eqqcolon t\left(\mu_{L}\right) \text{ .}\]
Directly, \[t'\left(\mu_{L}\right) = -2 \left[\left(1-a\right)\mu + \left(a-\mu\right) \mu_{L}\right] < 0 \text{ .}\]
Thus,
\[t\left(\mu_{L}\right) \leq t\left(\mu\right) = \left(1-\mu\right)\mu\left(1-\mu-a\right) \text{ ,}\] strictly if and only if $\mu_L > \mu$. The right-hand side of this is strictly positive if and only if $1-\mu-a > 0$. Moreover,
\[t\left(\mu_{L}\right) > t\left(1\right) = -a\left(1-\mu\right) < 0 \text{ .}\]

Accordingly, if $1 \leq \mu + a$, $\ubar{d}$ is always decreasing in $\mu_{L}$. If $1 > \mu + a$ then by the intermediate value theorem there exists some $\mu_{L}^{\dagger}$ such that $\ubar{d}$ is strictly increasing in $\mu_{L}$ if and only if $\mu_{L} \in \left(\mu, \mu_{L}^{\dagger}\right)$.
\end{proof}

\subsection{Theorem \ref{mainthm1} Proof}
\begin{proof}
A direct calculation reveals that $\lambda'\left(\mu_{L}\right)$ has the same sign as \[\left(\left(d-2a+1\right)\mu-ad+a\right)\mu_{L}+\left(a-1\right)\left(d+1\right)\mu \eqqcolon f\left(\mu_{L}\right) \text{ .}\]
Then,
\[f'\left(\mu_{L}\right) = \left(1-2a\right)\mu -\left(a-\mu\right)d+a \eqqcolon t\left(d\right) \text{ ,}\]
which is strictly decreasing in $d$ and hence \[t\left(d\right) > t\left(a\right) = \left(\mu+a\right)\left(1-a\right) > 0 \text{ .}\]
Thus, $f$ is strictly increasing in $\mu_{L}$. Moreover, \[f\left(1\right) = a\left(1-\mu\right)\left(1-d\right) > 0 \text{ .}\]
Next,
\[f\left(\mu\right) = -\mu \left(1-\mu\right) \left(1+d-2a\right) \text{ ,}\] which is strictly negative provided $a$ is sufficiently small ($2a < 1+d$). In this case, $\lambda$ is u-shaped in $\mu_L$. Moreover,
\[\lambda\left(\mu\right) = \left(\frac{1-a}{1-\mu}\right)\frac{\mu-a + d}{d}, \quad \text{and} \quad \lambda\left(1\right) = \frac{\mu}{a}\left(\frac{1-a}{1-\mu}\right) \text{ ;}\]
so, $\lambda\left(\mu\right) < \lambda\left(1\right)$ if and only if \[\left(a-\mu\right)\left(a-d\right) > 0 \text{ ,}\]
which holds by assumption. Naturally, if $2a \geq 1+d$, $f$ (and hence $\lambda'$) are always positive.
\end{proof}

\subsection{Proposition \ref{biasinc} Proof}

\begin{proof}
First, $\hat{\lambda}'\left(\mu_{DM}\right)$ has the same sign as
\[\left[\left(1-a_{DM}\right)p_{DM}'\left(a_{DM}\right) - p_{DM}\right]a_{DM}'\left(\mu_{DM}\right) \text{ .}\]
Second, it is easy to derive
\[\tag{$A3$}\label{exp57}a_{DM}'\left(\mu_{DM}\right) = -\frac{\left(1-\mu\right)\mu\left(1-a\right)a}{\left(\left(\mu - \mu_{DM}\right)a + \mu_{DM}\left(1-\mu\right)\right)^2} < 0 \text{ .}\]
Thus, it suffices to show that $p_{DM}'\left(a_{DM}\right) \leq 0$. 

Third, $p_{DM}'\left(a_{DM}\right)$ has the same sign as
\[\label{exa3}\tag{$A2$}-\left(\mu - b_{DM}\right) - \left(a_{DM}-\mu\right)b_{DM}'\left(a_{DM}\right) \text{ .}\]
Fourth, 
\[f'\left(b_{DM}\right) = \left(a_{DM}-b_{DM}\right)\phi''\left(b_{DM}\right) > 0 , \quad \text{and} \quad f'\left(a_{DM}\right) = - \left(\phi'\left(a_{DM}\right)-\phi'\left(b_{DM}\right)\right) < 0 \text{ ,}\]
so by the implicit function theorem, $b_{DM}'\left(a_{DM}\right) > 0$, and so Expression \ref{exa3} and $p_{DM}'\left(a_{DM}\right)$ are negative.
\end{proof}

\subsection{Proposition \ref{lastproposition} Proof}

\begin{proof}
The support of the optimal distribution over posteriors when the $DM$ is biased is $\left\{a_{DM}-d, a_{DM}\right\}$. A direct calculation reveals that 
\[\lambda'\left(\mu_{L}\right) \vert_{\mu_{L} = \mu} > \hat{\lambda}'\left(\mu_{DM}\right) \vert_{\mu_{DM} = \mu} \quad \Leftrightarrow \quad h\left(d\right) \coloneqq 4a^{2} - 3 a \mu - 2 a + \mu + \left(\mu - 2 a\right) d > 0 \text{ .}\]
Since $a > \mu$, $h$ is obviously decreasing in $d$. Thus, $h\left(d\right)$ is strictly positive if and only if 
\[d < \frac{4a^{2} - 3 a \mu - 2 a + \mu}{2a - \mu} \eqqcolon \bar{d}\text{ .}\]
Note that by assumption $d > a - \mu \ \left( > 0\right)$. Accordingly, we need 
\[a > \frac{1+\sqrt{1-2\mu+2\mu^{2}}}{2} \eqqcolon \bar{a} \text{ .}\]
By ``$a$ sufficiently high'' and ``$d$ sufficiently small'' in the proposition's statement, we mean $a > \bar{a}$ and $d < \bar{d}$. Next, $\lambda\left(\mu_B\right) - \hat{\lambda}\left(\mu_B\right)$ has the same sign as
\[\begin{split}
      f\left(\mu_B\right) &\coloneqq \mu_B\left(\left(\mu-a\right)\mu_B+\left(a-1\right)\mu\right)^2\left(\left(\mu^2+\left(d-1\right)\mu+\left(a-1\right)d\right)\mu_B-a\mu^2+\left(a-ad\right)\mu\right)\\ &-\mu\left(\left(\mu+a-1\right)\mu_B-a\mu\right)^2\left(\left(\mu-a\right)\mu_B^2+\left(d-1\right)\left(\mu-a\right)\mu_B+\left(a-1\right)d\mu\right) \text{ .}
\end{split}\]
Obviously, $f\left(\mu\right) = 0$. Moreover, 
\[f\left(a\right) = \left(a-1\right)a^3\left(\mu-a\right)\left(\left(\left(2a-1\right)\mu-a^2\right)d-\left(a-1\right)^2\mu\right) \eqqcolon g\left(d\right) \text{ .}\]
Next,
\[g'\left(d\right) = \left(1-a\right)a^3\left(a-\mu\right)\left(\left(2a-1\right)\mu-a^2\right) < 0 \text{ ,}\] since $a > \mu$. Thus, 
\[g\left(d\right) < g\left(0\right) = -\left(1-a\right)^3a^3\left(a-\mu\right)\mu < 0 \text{ ,}\]
so $f\left(a\right) < 0$. Directly, 
\[\frac{\partial{f}}{\partial{d}} = \underbrace{\left(\mu_{B}-\mu\right)}_{\geq 0}\underbrace{\left(\left(\mu-a\right)\mu_{B}+\left(a-1\right)\mu\right)}_{< 0}\underbrace{\left(\left(\mu-a\right)\mu_{B}-a\mu\right)}_{< 0}\underbrace{\left(\left(\mu+a-1\right)\mu_{B}-a\mu\right)}_{< 0} \leq 0 \text{ ,}\]
strictly so when $\mu_B > \mu$. Consequently, when $d \geq \bar{d}$ and $\mu_B > \mu$, $f\left(\mu_B, d\right) < f\left(\mu_B, \bar{d}\right)$, which has the same sign as
\[\begin{split}
    - \bigg(\left(\mu-a\right)&\left(\mu\left(a\left(\mu-a+2\right)-1\right)+2\left(a-1\right)a\left(2a-1\right)\right)\mu_B^2 +\left(2a^2-4a+1\right)\mu^2\left(\mu-a\right)\mu_B\\ &-\left(a-1\right)a\left(3a-1\right)\mu^3+2\left(a-1\right)a^2\left(2a-1\right)\mu^2\bigg)
\end{split} \text{ .}\]
This expression is strictly decreasing in $\mu_B$ and equals $0$ when $\mu_B = \mu$. Thus, when $d \geq \bar{d}$, $f$ is strictly negative for all $\mu_B \in \left(\mu,a\right)$. On the other hand, when $d < \bar{d}$, we need merely show that $f\left(\mu_B\right)$ has at most one real root in the interval $\left(\mu,a\right)$. But this is clear since $f$ is a quartic, $f\left(\mu\right) = 0$, $f'\left(\mu_B\right)\vert_{\mu = \mu_B} > 0$, $f\left(a\right) <0$, and $f\left(0\right) = a^2\left(1-a\right) d \mu^2 \geq 0$ (\textit{viz.}, $f$ would have to have three roots in $\left(\mu, a\right)$, which is impossible).
\end{proof}

\subsection{Theorem \ref{bigtheorem} Proof}
\begin{proof}
Let us begin with a sequence of claims. First,
\begin{claim}\label{newclaim1}
$b_L$ is strictly increasing in $\mu_L$.
\end{claim}
\begin{proof}
Recall that $b_L$ is implicitly defined as the root to $h$ given in Expression \ref{eq1}. Then, \[\tag{$A4$}\label{eqa7}b_{L}'\left(\mu_{L}\right) = - \frac{h'\left(\mu_{L}\right)}{h'\left(b_{L}\right)} = \frac{\left(\phi'\left(a_L\right) - \phi'\left(b_{L}\right)\right)a_L'\left(\mu\right)}{\left(a_L-b_{L}\right)\phi''\left(b_{L}\right)} > 0 \text{ ,}\]
since $a_L'\left(\mu_L\right) > 0$.
\end{proof}
Second,
\begin{claim}\label{newclaim2}
$b_L$ is strictly decreasing in $v$.
\end{claim}
\begin{proof}
Directly, $h'\left(v\right) = 1$, so we may again appeal to the implicit function theorem to get the result.
\end{proof}
Third,
\begin{claim}\label{newclaim3}
$b_{DM}$ is strictly decreasing in $\mu_{DM}$.
\end{claim}
\begin{proof}
Recall that $b_{DM}$ is implicitly defined as the root to $f$ given in Expression \ref{eq2n2}. Then,
\[\tag{$A5$}\label{eqa8}b_{DM}'\left(\mu_{DM}\right) = - \frac{f'\left(\mu_{DM}\right)}{f'\left(b_{DM}\right)} = \frac{\left(\phi'\left(a_{DM}\right) - \phi'\left(b_{DM}\right)\right)a_{DM}'\left(\mu_{DM}\right)}{\left(a_{DM}-b_{DM}\right)\phi''\left(b_{DM}\right)} < 0 \text{ ,}\]
since $a_{DM}'\left(\mu_{DM}\right) < 0$ (see Expression \ref{exp57}).
\end{proof}
Fourth, 
\begin{claim}\label{newclaim4}
$b_{DM}$ is strictly decreasing in $v$.
\end{claim}
\begin{proof}
Directly, $f'\left(v\right) = 1$, so we may again appeal to the implicit function theorem to get the result.
\end{proof}
Using Claim \ref{newclaim2}, it is easy to see that for any $\mu_L$ and $a$, there exists an interior solution if and only if $v$ is sufficiently high. The same is true when the $DM$ is biased (using Claim \ref{newclaim4}). We define $v_{\epsilon}$ to be this minimal payoff when $a = 1-\epsilon$ and $\mu_B = \mu$. Formally,
\[v_{\epsilon} \coloneqq \inf\left\{v\colon b < \mu, a = 1-\epsilon, \mu_B = \mu, \ b \ \text{is the solution to} \ h\left(b,a,v\right) = 0\right\} \text{ .}\]
Next,
\begin{claim}\label{newclaim5}
For all $v > v_{\epsilon}$, the solutions are interior with $b_L < \mu_L$ and $b_{DM} < \mu$ for all $\mu_B$ sufficiently close to $\mu$.
\end{claim}
\begin{proof}
Let $v > v_{\epsilon}$. By the continuity of $b_L$ in $\mu_L$, for any $\eta \in \left(0,\mu-b\right)$, there exists a $\delta > 0$ such that if $\mu_L \leq \mu + \delta$, $b_L \leq b + \eta$. By definition $b < \mu$, so let us set $\eta \coloneqq \left(\mu-b\right)/2$. Thus, we get $b_L \leq b + \left(\mu-b\right)/2 < \mu < \mu_L$.

Claim \ref{newclaim3} immediately implies the second part of this claim.
\end{proof}
Next, explicitly, \[\label{a1a}\tag{$A1$}\begin{split}
    \lambda'\left(\mu_{L}\right) &= \left(\frac{1 - a_{L} - \left(1-\mu_{L}\right)a_{L}'\left(\mu_{L}\right)}{\left(1-\mu_{L}\right)^2}\right)p_{L} + \frac{1-a_{L}}{1-\mu_{L}}p_{L}'\left(\mu_{L}\right)\\
    &= \underbrace{-\frac{\left(1-a\right)\mu\left(a-\mu\right)}{\left(\left(\mu_{L} - \mu\right)a + \mu\left(1-\mu_{L}\right)\right)^2}p_{L}}_{ < 0} + \frac{\left(1-a\right)\mu}{\left(\mu_{L} - \mu\right)a + \mu\left(1-\mu_{L}\right)}p_{L}'\left(\mu_{L}\right)
\end{split} \text{ .}\]
By definition, when $v = v_\epsilon$, $b = \mu$. Accordingly, from Expression \ref{a1a}, when $v = v_\epsilon$
\[\lambda'\left(\mu_{L}\right)\vert_{\mu_{L}=\mu} = \frac{1-a}{1-\mu}p_{L}'\left(\mu_{L}\right)\vert_{\mu_{L}=\mu} \text{ ;}\]
that is, the first term vanishes. Analogously, 
\[\hat{\lambda}'\left(\mu_{DM}\right)\vert_{\mu_{DM}=\mu} = \frac{1-a}{1-\mu}p_{DM}'\left(\mu_{DM}\right)\vert_{\mu_{DM}=\mu} \text{ .}\]
Directly,
\[\begin{split}
    p_{L}'\left(\mu_{L}\right)\vert_{\mu_{L} = \mu} &= \frac{a\left(1-b_{L}'\left(\mu_L\right)\vert_{\mu_L = \mu}\right)-b\left(1-a_{L}'\left(\mu_L\right)\vert_{\mu_L = \mu}\right) - \mu_L\left(a_{L}'\left(\mu_L\right)\vert_{\mu_L = \mu}-b_{L}'\left(\mu_L\right)\vert_{\mu_L = \mu}\right)}{\left(a-b\right)^2}\\
    &= \frac{1-b_{L}'\left(\mu_L\right)\vert_{\mu_L = \mu}}{a-\mu}\\
\end{split}\text{ ,}\]
when $v = v_e$. Similarly, 
\[\begin{split}
    p_{DM}'\left(\mu_{DM}\right)\vert_{\mu_{DM} = \mu} &= \frac{-a b_{DM}'\left(\mu_{DM}\right)\vert_{\mu_{DM} = \mu} + b a_{DM}'\left(\mu_{DM}\right)\vert_{\mu_{DM} = \mu} - \mu\left(a_{DM}'\left(\mu_{DM}\right)\vert_{\mu_{DM} = \mu}-b_{DM}'\left(\mu_{DM}\right)\vert_{\mu_{DM} = \mu}\right)}{\left(a-b\right)^2}\\
    &= -\frac{b_{DM}'\left(\mu_{DM}\right)\vert_{\mu_{DM} = \mu}}{a-\mu}
\end{split} \text{ ,}\]
when $v = v_e$. Since $a < 1$, $\lambda'\left(\mu_{L}\right)\vert_{\mu_{L}=\mu}-\hat{\lambda}'\left(\mu_{DM}\right)\vert_{\mu_{DM}=\mu}$ has the same sign as $p_{L}'\left(\mu_{L}\right)\vert_{\mu_{L}=\mu} - p_{DM}'\left(\mu_{DM}\right)\vert_{\mu_{DM}=\mu}$, which has the same sign as \[1-b_{L}'\left(\mu_L\right)\vert_{\mu_L = \mu}+b_{DM}'\left(\mu_{DM}\right)\vert_{\mu_{DM} = \mu} \text{ ,}\]
which equals, using Expression \ref{exp57} and Equations \ref{eqa7} and \ref{eqa8},
\[\label{eqa9}\tag{$A7$}1 - \underbrace{2\left(1-a\right) a\frac{\phi'\left(a\right) - \phi'\left(\mu\right)}{\left(a- \mu\right)\mu\left(1-\mu\right) \phi''\left(\mu\right)}}_{q\left(a\right)} \text{ .}\]
If $\phi'$ is bounded, then so is $\phi''$. Thus, as $a \nearrow 1$, $q\left(a\right) \searrow 0$. On the other hand, if $\phi = \phi_{LL}$ Expression \ref{eqa9} becomes
\[1 - \underbrace{2\left(1-a\right) a\mu\left(1-\mu\right)\frac{\phi_{LL}'\left(a\right) - \phi_{LL}'\left(\mu\right)}{\left(a- \mu\right)}}_{q_{LL}\left(a\right)} \text{ .}\]
As $a \nearrow 1$, $q_{LL} \searrow 2\mu$. Finally, if $\phi = \phi_{T}$, Expression \ref{eqa9} becomes 
\[1 - \underbrace{2\left(1-a\right) a\frac{\phi_{T}'\left(a\right) - \phi_{T}'\left(\mu\right)}{\left(a- \mu\right)}}_{q_{T}\left(a\right)} \text{ .}\]
As $a \nearrow 1$, $q_{T} \searrow 0$.

The result then follows from the continuity of the solution in the parameters.
\end{proof}

\subsection{Proposition \ref{firstpreference} Proof}
\begin{proof}
Here, the conditional probabilities are
\[\gamma = \frac{a}{\mu} p, \quad \text{and} \quad \lambda = \frac{1-a}{1-\mu} p, \quad \text{where} \quad p = \frac{\mu- b}{a - b} \text{ .}\]
Directly, $\lambda'\left(a\right)$ has the same sign as
\[-p + \left(1-a\right)p'\left(a\right) \text{ .}\]
Consequently, we need to show that $p'\left(a\right) \leq 0$. As in the proof of Proposition \ref{biasinc}, it suffices to show that $b'\left(a\right) \geq 0$. To that end, we compute
\[s'\left(b\right) = \left(a-b\right) \phi''\left(b\right) > 0 \text{ ,}\]
and 
\[s'\left(a\right) = \eta\left(\rho+1\right) - \left(\phi'\left(a\right)-\phi'\left(b\right)\right) \leq 0 \text{ ,}\]
by Assumption \ref{ass2}. Finally, by the implicit function theorem, $b'\left(a\right) \geq 0$.
\end{proof}

\subsection{Proposition \ref{finalpropprop}}
\begin{proof}
$\lambda'\left(v\right)$ has the same sign as $p'\left(v\right)$, which has the same sign as $-b'\left(v\right)$. Analogously, $\lambda'\left(\rho\right)$ has the same sign as $-b'\left(\rho\right)$ and $\lambda'\left(\eta\right)$ has the same sign as $-b'\left(\eta\right)$. From the proof of Proposition \ref{firstpreference}, $s'\left(b\right) > 0$. Directly, $s'\left(v\right) = 1- \eta \geq 0$ and $s'\left(\rho\right) = -\eta \left(1-a\right) \leq 0$. Accordingly, the implicit function theorem yields that $b'\left(v\right) \leq 0$ and $b'\left(\rho\right) \geq 0$.

Directly,
\[s'\left(\eta\right) = a - \rho \left(1-a\right) - v \text{ ,}\]
and appealing to the implicit function theorem, $s'\left(\eta\right)$ has the same sign as $\lambda'\left(\eta\right)$.
\end{proof}
\end{document}